\documentclass[a4paper,12pt]{article}

\usepackage{amssymb,amsmath,amsthm}

\usepackage[english]{babel}

\usepackage[margin=2cm]{geometry}
\usepackage{setspace} 
\usepackage{mathptmx} 
\usepackage{longtable} 

\usepackage[unicode=true, bookmarks=false, breaklinks=false,
pdfborder={0 0 1},backref=false,colorlinks=true]{hyperref}
\hypersetup{xetex,linkcolor=blue,citecolor=[rgb]{0,0.6,0},urlcolor=blue}

\makeatletter

\theoremstyle{plain}
\newtheorem{thm}{\protect\theoremname}[section]
  \theoremstyle{definition}
  \newtheorem{defn}[thm]{\protect\definitionname}
  \theoremstyle{plain}
  \newtheorem*{thm*}{\protect\theoremname}
  \theoremstyle{plain}
  
  \theoremstyle{plain}
  \newtheorem{lem}[thm]{\protect\lemmaname}
  \theoremstyle{remark}
  \newtheorem{rem}[thm]{\protect\remarkname}
  \theoremstyle{plain}
  \newtheorem{cor}[thm]{\protect\corollaryname}

\makeatother

  \providecommand{\corollaryname}{Corollary}
  \providecommand{\definitionname}{Definition}
  \providecommand{\lemmaname}{Lemma}
  \providecommand{\propositionname}{Proposition}
  \providecommand{\remarkname}{Remark}
  \providecommand{\theoremname}{Theorem}
\providecommand{\theoremname}{Theorem}

\begin{document}

\title{Random Schr\"odinger operators with\linebreak{}
a background potential}

\author{Hayk Asatryan and Werner Kirsch\\
Faculty of Mathematics and Computer Science\\
FernUniversit\"at in Hagen, Germany\smallskip{}
}

\maketitle
{\small{}}%
\begin{longtable}{p{12cm}}
\textbf{\small{}Abstract.}{\small{} We consider one-dimensional random
Schr\"odinger operators with a background potential, arising in the
inverse problem of scattering. We study the influence of the background
potential on the essential spectrum of the random Schr\"odinger operator
and obtain Anderson Localization for a larger class of one-dimensional
Schr\"odinger operators. Further, we prove the existence of the integrated
density of states and give a formula for it.}\tabularnewline
\end{longtable}{\small \par}

\bigskip{}

\section{Introduction}

The fundamental work of physicist P. W. Anderson \cite{And58} gave
rise to a lot of further investigations of different mathematicians.
P. W. Anderson argued that electrons inside a disordered quantum mechanical
system are localized (named Anderson Localization lately), provided
that the degree of randomness of the impurities or defects is sufficiently
large. One way to express Anderson Localization in mathematically
rigorous terms is that the corresponding Schr\"odinger (Hamilton) operator
has a pure point spectrum. For one-dimensional Schr\"odinger operators
Anderson expected localization for all energies and arbitrary small
disorder.

In the one-dimensional model the Schr\"odinger operator is the alloy-type
operator 
\[
-\frac{d^{2}}{dx^{2}}+V_{\omega}\:,
\]
where the random potential $V_{\omega}$ may for example be of the
form 
\begin{equation}
V_{\omega}(x):=\sum_{k=-\infty}^{\infty}q_{k}(\omega)f(x-k)\quad\left(x\in\mathbb{R}\right).\label{eq:0010}
\end{equation}
Here $f$ is a real-valued function and $q_{k}\;\left(k\in\mathbb{Z}\right)$
are independent random variables with a common distribution $P_{0}$.
The first proof of Anderson Localization for a related one-dimensional
model was given by I. Goldsheid, S. Molchanov and L. Pastur in \cite{Gold77}.
For the discrete analogue of \eqref{eq:0010} localization was first
proved by H. Kunz and B. Souillard \cite{Kunz80}. The first complete
proofs of localization in the higher dimensional case were given in
\cite{FrSp83}, \cite{AizMol93}. The existence of the integrated
density of states for alloy-type operators was established in \cite{Pas73}
and \cite{KiMados}.

On the space $L^{2}(\mathbb{R})$ we consider Schr\"odinger operators
of the form 
\begin{align*}
H_{\omega}:=-\frac{d^{2}}{dx^{2}}+U+V_{\text{per}}+V_{\omega}.
\end{align*}
We assume that the background potential $U$ belongs to the space
of real-valued uniformly locally square-integrable functions 
\begin{align*}
L_{\text{unif }}^{2,\,\text{loc}}:=\{F:\mathbb{R}\to\mathbb{R}\mid\sup_{x\in\mathbb{R}}\;\int\limits _{x-1}^{x+1}|F(x)|^{2}dx<\infty\}
\end{align*}
and satisfies the relations 
\begin{align*}
U(x)~\to~a^{-}\quad\text{as}\quad x\to-\infty,\qquad U(x)~\to~a^{+}\quad\text{as}\quad x\to+\infty
\end{align*}
with $a^{\pm}\in\mathbb{R}$. The potential $U$ arises in the inverse
problem of scattering (see \cite{Asat05}). $V_{\text{per}}$ is assumed
to be a $1$-periodic real-valued function in $L_{\text{unif }}^{2,\,\text{loc}}$,
and $V_{\omega}$ is a random alloy-type potential of the form \eqref{eq:0010}.
We suppose that $f$, called the single-site potential, satisfies
the estimate 
\begin{align*}
\left|f(x)\right|\leqslant C\left(1+\left|x\right|\right)^{-\gamma}\quad\left(x\in\mathbb{R}\right)
\end{align*}
for some $\gamma>1$.

We assume for simplicity that $\operatorname{supp}P_{0}$ is a compact
subset of $\mathbb{R}$. We remark that the existence of enough moments
of $P_{0}$ would be sufficient. Moreover, $f$ may have local singularities.

Under the above assumptions, the potentials $U$, $V_{\text{per}}$,
$V_{\omega}$ and their sums belong to $L_{\text{unif }}^{2,\,\text{loc}}$,
hence they are $H_{0}$-bounded (see \cite{ReedSimon4}, Theorem XIII.96).
Moreover, the operators 
\begin{align*}
 & H_{0}:=~-\frac{d^{2}}{dx^{2}}\quad\mbox{(the free Hamiltonian)},\\
 & H_{U}:=~H_{0}+U,\\
 & H_{\text{per}}:=~H_{0}+V_{\text{per}},\\
 & H_{U,\,\text{per}}:=~H_{0}+U+V_{\text{per}}
\end{align*}
and $H_{\omega}$ are essentially self-adjoint on $C_{0}^{\infty}(\mathbb{R})$.

Throughout this work, $\left\Vert \cdot\right\Vert $ will denote
the $L^{2}$-norm. The spectrum and the essential spectrum of a linear
operator $A$ will be denoted by $\sigma\left(A\right)$ and $\sigma_{\text{ess}}\left(A\right)$,
respectively.

\section{The essential spectra of $H_{U+V}$ and $H_{U,\,\text{per}}$}

One of the main observations of this section is the following result.

\begin{thm}\label{th:ess1} Let $U_{1},U_{2},V:\mathbb{R}\to\mathbb{R}$
be $H_{0}$-bounded measurable functions and 
\[
U_{j}(x)\xrightarrow[x\to-\infty]{}a^{-},\quad U_{j}(x)\xrightarrow[x\to\infty]{}a^{+}\quad\left(j=1,2\right)
\]
for some $a^{\pm}\in\mathbb{R}$. Then 
\[
\sigma_{\text{\textup{ess}}}\left(H_{U_{1}+V}\right)=\sigma_{\text{\textup{ess}}}\left(H_{U_{2}+V}\right).
\]
\end{thm}

\begin{proof} We need to show that 
\[
\sigma_{\text{ess}}\left(H_{U_{1}+V}\right)\subset\sigma_{\text{ess}}\left(H_{U_{2}+V}\right),
\]
\[
\sigma_{\text{ess}}\left(H_{U_{2}+V}\right)\subset\sigma_{\text{ess}}\left(H_{U_{1}+V}\right).
\]
We'll prove the first inclusion (the proof of the second one is similar).
Let $\lambda\in\sigma_{\text{ess}}\left(H_{U_{1}+V}\right)$. By Weyl's
criterion and Theorem 3.11 in \cite{Cyc08} we conclude that there
is a Weyl sequence of functions $\varphi_{n}\in C_{0}^{\infty}(\mathbb{R})\;\:\left(n\in\mathbb{N}\right)$
such that

\[
\left\Vert \varphi_{n}\right\Vert =1\quad\left(n\in\mathbb{N}\right),
\]
\begin{equation}
\left\Vert \left(H_{U_{1}+V}-\lambda I\right)\varphi_{n}\right\Vert \to0\label{eq:1099}
\end{equation}
and either 
\begin{align}
\operatorname{supp}\varphi_{n}~\subset(-\infty,-n)\quad & \text{for all \ensuremath{\ensuremath{n\in\mathbb{N}}}}\label{eq:minus}
\end{align}
or 
\begin{align*}
\operatorname{supp}\varphi_{n}~\subset(n,\infty)\quad & \text{for all \ensuremath{n\in\mathbb{N}}}
\end{align*}
holds. Assume \eqref{eq:minus} is true, then 
\[
\left\Vert \left(H_{U_{1}+V}-\lambda I\right)\varphi_{n}\right\Vert -\left\Vert \left(H_{V}-\left(\lambda-a^{-}\right)I\right)\varphi_{n}\right\Vert \to0,
\]
\[
\left\Vert \left(H_{U_{2}+V}-\lambda I\right)\varphi_{n}\right\Vert -\left\Vert \left(H_{V}-\left(\lambda-a^{-}\right)I\right)\varphi_{n}\right\Vert \to0
\]
and hence 
\[
\left\Vert \left(H_{U_{1}+V}-\lambda I\right)\varphi_{n}\right\Vert -\left\Vert \left(H_{U_{2}+V}-\lambda I\right)\varphi_{n}\right\Vert \to0.
\]
From this and \eqref{eq:1099} we obtain 
\[
\left\Vert \left(H_{U_{2}+V}-\lambda I\right)\varphi_{n}\right\Vert \to0,
\]
therefore $\lambda\in\sigma_{\text{ess}}\left(H_{U_{2}+V}\right)$.
\end{proof}

As a corollary to the proof of Theorem \ref{th:ess1} we get \begin{cor}
\label{thm1040} Let $U,V:\mathbb{R}\to\mathbb{R}$ be measurable,
$H_{0}$-bounded and 
\[
U(x)\xrightarrow[x\to-\infty]{}a^{-},\quad U(x)\xrightarrow[x\to\infty]{}a^{+}
\]
(in the usual sense), where $a^{\pm}\in\mathbb{R}$. Then 
\begin{equation}
\sigma_{\text{\textup{ess}}}\left(H_{U+V}\right)\subset\left(a^{-}+\sigma_{\text{\textup{ess}}}\left(H_{V}\right)\right)\cup\left(a^{+}+\sigma_{\text{\textup{ess}}}\left(H_{V}\right)\right).\label{eq:1100}
\end{equation}
\end{cor}

\begin{rem} The previous theorem shows that the knowledge of $V$
and $a^{\pm}$ is sufficient for the unique determination of $\sigma_{\text{ess}}\left(H_{U+V}\right)$.
In fact, 
\[
\sigma_{\text{ess}}\left(H_{U+V}\right)=\sigma_{\text{ess}}\left(H_{U_{c}+V}\right),
\]
where $U_{c}=a^{-}\chi_{(-\infty,0]}+a^{+}\chi_{(0,\infty)}$. \end{rem}
In general, equality in \eqref{eq:1100} does not hold. However, for
the case of periodic potentials we have:

\begin{thm} \label{thm1050} Let $U:\mathbb{R}\to\mathbb{R}$ be
measurable, $H_{0}$-bounded and satisfy the conditions 
\[
U(x)\xrightarrow[x\to-\infty]{}a^{-},\quad U(x)\xrightarrow[x\to\infty]{}a^{+},
\]
and let $V_{\text{\textup{per}}}$ be a $H_{0}$-bounded periodic
potential, then 
\begin{equation}
\sigma_{\text{\textup{ess}}}\left(H_{U,\,\text{\textup{per}}}\right)=\left(a^{-}+\sigma_{\text{\textup{ess}}}\left(H_{\text{\textup{per}}}\right)\right)\cup\left(a^{+}+\sigma_{\text{\textup{ess}}}\left(H_{\text{\textup{per}}}\right)\right).\label{eq:1105}
\end{equation}
\end{thm}

\begin{proof}In the view of Corollary \ref{thm1040}, we need to
prove that 
\begin{equation}
a^{-}+\sigma_{\text{ess}}\left(H_{\text{per}}\right)\subset\sigma_{\text{ess}}\left(H_{U,\,\text{per}}\right),\label{eq:1110}
\end{equation}
\begin{equation}
a^{+}+\sigma_{\text{ess}}\left(H_{\text{per}}\right)\subset\sigma_{\text{ess}}\left(H_{U,\,\text{per}}\right).\label{eq:1120}
\end{equation}
We'll prove \eqref{eq:1110} (the proof of \eqref{eq:1120} is similar).
Let $\lambda\in a^{-}+\sigma_{\text{ess}}\left(H_{\text{per}}\right)$,
i.e., $\lambda-a^{-}\in\sigma_{\text{ess}}\left(H_{\text{per}}\right)$.
Then there is a Weyl sequence $\varphi_{n}\in C_{0}^{\infty}(\mathbb{R})\;\:\left(n\in\mathbb{N}\right)$
with 
\begin{enumerate}
\item $\left\Vert \varphi_{n}\right\Vert =1\;\:\left(n\in\mathbb{N}\right)$, 
\item $\left\Vert \left(H_{\text{per}}-(\lambda-a^{-})I\right)\varphi_{n}\right\Vert \to0$, 
\end{enumerate}
Since $V_{\text{per}}$ is periodic, any shift of $\varphi_{n}$ by
an integer is also a Weyl sequence for $H_{\text{per}}+a^{-}$. Thus
we may assume that $\operatorname{supp}\varphi_{n}\subset(-\infty,-n)\;\:\left(n\in\mathbb{N}\right)$.
As in the previous proofs, one easily sees that this sequence is also
a Weyl sequence for $H_{\text{per}}+U$. \end{proof}

\begin{rem} It is well known that under the above assumptions on
$V_{\text{per}}$ the equality\linebreak{}
 $\sigma_{\text{ess}}\left(H_{\text{per}}\right)=\sigma\left(H_{\text{per}}\right)$
holds (see \cite{East73}, \cite{ReedSimon4}). \end{rem}

\begin{rem} The special case of the formula \eqref{eq:1105} in which
$V_{\text{per}}=0$, yields 
\[
\sigma_{\text{ess}}\left(H_{U}\right)=\left[\min\left\{ a^{+},a^{-}\right\} ,\infty\right).
\]
This equality was obtained by I. Khachatryan and A. Petrosyan \cite{Khach04}
under the condition 
\[
\int\limits _{-\infty}^{0}\left|U\left(x\right)-a^{-}\right|dx+\int\limits _{0}^{\infty}\left|U\left(x\right)-a^{+}\right|dx<\infty
\]
(see also \cite{Asat05}). \end{rem}

\section{The essential spectrum of $H_{\omega}$}

We turn to the spectrum of $H_{\omega}$. To do so, we first describe
the spectrum of $H_{\text{per}}+V_{\omega}$, i.e., the case $U=0$.
We follow the investigation in \cite{KiMa}. \begin{defn} A potential
$W(x)=\sum\limits _{k\in\mathbb{Z}}\,\rho_{k}\,f(x-k)$ is called
\emph{admissible}, if $\rho_{k}\in\operatorname{supp}P_{0}$ for all
$k$. Let us denote by $\mathcal{P}$ the set of all admissible potentials,
generated by $\ell$-periodic $\rho_{k}$ for some $\ell\in\mathbb{N}$.
\end{defn}

\begin{thm}\label{th:KiMa} The spectrum $\sigma\left(H_{\text{\textup{per}}}+V_{\omega}\right)$
is almost surely independent of $\omega$ and is (almost surely) given
by 
\begin{align*}
\sigma\left(H_{\text{\textup{per}}}+V_{\omega}\right)=\sigma_{\text{\textup{ess}}}\left(H_{\text{\textup{per}}}+V_{\omega}\right)=\,\overline{\bigcup_{W\in\mathcal{P}}\,\sigma\left(H_{\text{\textup{per}}}+W\right)}\,.
\end{align*}
\end{thm} In the case of $V_{\text{per}}=0$ Theorem \ref{th:KiMa}
is proved in \cite{KiMa}; the proof in the general case is similar.

In particular, the following result was proved in \cite{KiMa}.

\begin{lem}\label{lem:KiMa} If $W$ is a periodic admissible potential
and $\lambda\in\sigma(H_{\text{\textup{per}}}+W)$, then there are
sequences $\varphi_{n}^{+},\varphi_{n}^{-}\in L^{2}(\mathbb{R})\;\:\left(n\in\mathbb{N}\right)$
in the domain of $H_{\text{\textup{per}}}+W$, such that 
\begin{enumerate}
\item $\|\varphi_{n}^{+}\|~=~\|\varphi_{n}^{-}\|~=~1\quad\left(n\in\mathbb{N}\right).$ 
\item The supports of $\varphi_{n}^{+}$ and $\varphi_{n}^{-}$ are compact
and satisfy

$\operatorname{supp}\varphi_{n}^{+}~\subset~[n,\infty)\quad$ and
$\quad\operatorname{supp}\varphi_{n}^{-}~\subset~(-\infty,-n]\quad\left(n\in\mathbb{N}\right).$

\item For almost all $\omega$

$\|\left(H_{\text{\textup{per}}}+V_{\omega}-\lambda\right)\varphi_{n}^{+}\|~\to~0\quad$
and $\quad\|\left(H_{\text{\textup{per}}}+V_{\omega}-\lambda\right)\varphi_{n}^{-}\|~\to~0.$

\end{enumerate}
\end{lem} \goodbreak

From this we conclude

\begin{thm} Almost surely 
\begin{align*}
\sigma_{\text{\textup{ess}}}\left(H_{\text{\textup{per}}}+U+V_{\omega}\right)=\left(a^{-}+\sigma_{\text{\textup{ess}}}\left(H_{\text{\textup{per}}}+V_{\omega}\right)\right)\cup\left(a^{+}+\sigma_{\text{\textup{ess}}}\left(H_{\text{\textup{per}}}+V_{\omega}\right)\right).
\end{align*}
\end{thm}

\begin{proof} By Corollary \ref{thm1040} we know that 
\begin{align*}
\sigma_{\text{ess}}\left(H_{\text{\textup{per}}}+U+V_{\omega}\right)\subset\left(a^{-}+\sigma_{\text{ess}}\left(H_{\text{\textup{per}}}+V_{\omega}\right)\right)\cup\left(a^{+}+\sigma_{\text{ess}}\left(H_{\text{\textup{per}}}+V_{\omega}\right)\right).
\end{align*}
To prove the converse we observe that for any $W\in\mathcal{P}$ 
\begin{align*}
a^{\pm}+\sigma\left(H_{\text{\textup{per}}}+W\right)\subset\sigma_{\text{ess}}\left(H_{\text{\textup{per}}}+U+W\right)
\end{align*}
by Theorem \ref{thm1050}. It is easy to see (e.g. as in \cite{KiMa})
that almost surely for $W\in\mathcal{P}$ 
\begin{align*}
\sigma_{\text{ess}}\left(H_{\text{\textup{per}}}+U+W\right)\subset\sigma_{\text{ess}}\left(H_{\text{\textup{per}}}+U+V_{\omega}\right).
\end{align*}
Hence we conclude that 
\begin{align*}
\bigcup_{W\in\mathcal{P}}\sigma\left(H_{\text{\textup{per}}}+W+a^{+}\right)\cup\bigcup_{W\in\mathcal{P}}\sigma\left(H_{\text{\textup{per}}}+W+a^{-}\right)\subset\sigma_{\text{ess}}\left(H_{\text{\textup{per}}}+U+V_{\omega}\right).
\end{align*}
Since the right side is a closed set, we infer from Theorem \ref{th:KiMa}
that almost surely 
\begin{align*}
\left(a^{-}+\sigma_{\text{ess}}\left(H_{\text{\textup{per}}}+V_{\omega}\right)\right)\cup\left(a^{+}+\sigma_{\text{ess}}\left(H_{\text{\textup{per}}}+V_{\omega}\right)\right)\subset\sigma_{\text{ess}}\left(H_{\text{\textup{per}}}+U+V_{\omega}\right)\,.
\end{align*}
\end{proof}

The following localization result is based on the work of D. Damanik
and G. Stolz \cite{DamSt11}.

\begin{thm} Let $U$ be continuous, $V_{\text{\textup{per}}}$ be
bounded and $f$ satisfy the estimate 
\[
c\chi_{I}\left(x\right)\leqslant f\left(x\right)\leqslant C\chi_{\left(0,1\right)}\quad\left(\mbox{a.e.}\;x\in\mathbb{R}\right)
\]
with constants $0<c\leqslant C<\infty$ and a non-trivial subinterval
$I$ of $\left(0,1\right)$. Moreover, let 
\[
f\left(x\right)>0\quad\left(\mbox{a.e.}\;x\in\left(a,b\right)\right),
\]
\[
f\left(x\right)=0\quad\left(\mbox{a.e.}\;x\in\mathbb{R}\backslash\left(a,b\right)\right)
\]
for a subinterval $\left(a,b\right)\subset\left(0,1\right)$. Then
almost surely the operator $H_{\omega}$ has a dense point spectrum
with exponentially decaying eigenfunctions and, possibly in addition,
isolated eigenvalues with finite multiplicities. \end{thm} \emph{Proof.}
Damanik and Stolz \cite{DamSt11} proved that (under the formulated
assumptions on $f$) if for a function $W_{0}\in L_{\mathbb{R}}^{\infty}\left(\mathbb{R}\right)$
the set 
\[
M\left(W_{0}\right):=\left\{ W_{0}\left(\cdot-n\right)\big|_{\left(0,1\right)}:\;n\in\mathbb{Z}\right\} 
\]
is relatively compact in $L^{\infty}\left(0,1\right)$, then almost
surely the operator $-\frac{d^{2}}{dx^{2}}+W_{0}+V_{\omega}$ has
a dense point spectrum with exponentially decaying eigenfunctions
and, possibly in addition, isolated eigenvalues with finite multiplicities.
Since 
\[
M\left(U+V_{\text{per}}\right)=M\left(U\right)+V_{\text{per}}\,,
\]
hence it remains to show the relative compactness of $M\left(U\right)$.
The continuity of $U$ and the existence of its finite limits at $\pm\infty$
imply the uniform continuity of $U$. The latter, in turn, implies
the equicontinuity of $M\left(U\right)$. According to Arzela-Ascoli
theorem, $M\left(U\right)$ is relatively compact in $C\left[0,1\right]$
and hence in $L^{\infty}\left(0,1\right)$.\qed

\section{The Integrated Density of States}

In this section we investigate the integrated density of states of
the operator $H_{\omega}$.

\begin{defn} Let $A$ be a self-adjoint operator bounded below and
with (possibly infinite) purely discrete spectrum $\lambda_{1}(A)\leqslant\lambda_{2}(A)\leqslant\lambda_{3}(A)\leqslant\ldots,$
where the eigenvalues are counted according to their multiplicities.
Denote 
\[
N(A,E):=\#\,\{j:\;\lambda_{j}(A)\leqslant E\}\quad\left(E\in\mathbb{R}\right).
\]
For $H=H_{0}+W$ with $W\in L_{\text{unif }}^{2,\,\text{loc}}$ and
$a,b\in\mathbb{R},\:a<b$ we define $H_{a,b}^{D}$ to be the operator
$H$ restricted to $L^{2}(a,b)$ with Dirichlet boundary conditions
both at $a$ and $b$. Similarly, $H_{a,b}^{N}$ has Neumann boundary
conditions at $a$ and $b$, $H_{a,b}^{D,N}$ has Dirichlet boundary
condition at $a$ and Neumann boundary condition at $b$, $H_{a,b}^{N,D}$
has Neumann boundary condition at $a$ and Dirichlet one at $b$.
\medskip{}

If for $H=H_{0}+W$ the limit 
\begin{align*}
\mathcal{N}(E)=\mathcal{N}(H,E):=\lim_{L\to\infty}\,\frac{1}{2L}\,N\left(H_{-L,L}^{D},E\right)
\end{align*}
exists for all but countably many $E$, we call $\mathcal{N}(E)$
the \emph{integrated density of states} for $H$. \end{defn}

It is well known that under our assumptions the integrated density
of states for $H_{\text{per}}+V_{\omega}$ exists, more precisely:
\begin{thm}\label{th:DOSexist} If $V_{\omega}$ satisfies the assumptions
of Section 1, then the integrated density of states $\mathcal{N}(H_{\text{\textup{per}}}+V_{\omega},E)$
almost surely exists and for all but countably many $E$ the following
equalities hold: 
\[
\mathcal{N}(H_{\text{\textup{per}}}+V_{\omega},E)=\lim_{L\to\infty}\,\frac{N\left(H_{-L,L}^{N}\left(E\right)\right)}{2L}=\lim_{L\to\infty}\,\frac{\mathbb{E}\left(N\left(H_{-L,L}^{D}\left(E\right)\right)\right)}{2L}=\lim_{L\to\infty}\,\frac{\mathbb{E}\left(N\left(H_{-L,L}^{N}\left(E\right)\right)\right)}{2L}
\]
($\mathbb{E}$ denotes the expectation with respect to $\mathbb{P}$).
\end{thm} In the case of $V_{\text{per}}=0$ Theorem \ref{th:DOSexist}
is proved in \cite{KiMados}; the proof in the general case is similar
and uses the method of Dirichlet-Neumann bracketing (see \cite{ReedSimon4}).
In particular, it is used:

\begin{thm}\label{th:DNbrack} If $a<c<b$ and $X,Y\in\{D,N\}$,
then 
\[
N\left(H_{a,c}^{X,D},E\right)+N\left(H_{c,b}^{D,Y},E\right)\leqslant N\left(H_{a,b}^{X,Y},E\right)\leqslant N\left(H_{a,c}^{X,N},E\right)+N\left(H_{c,b}^{N,Y},E\right)\quad\left(E\in\mathbb{R}\right).
\]
\end{thm}

For the integrated density of states of the operator $H_{\omega}$
we have the following result.

\begin{thm}\label{th:DOScomp} The integrated density of states $\mathcal{N}(H_{\omega},E)$
almost surely exists and can be expressed in terms of $\mathcal{N}_{1}(E)$,
the integrated density of states of $H_{\text{\textup{per}}}+V_{\omega}$
by: 
\[
\mathcal{N}(H_{\omega},E)~=~\frac{1}{2}\,\mathcal{N}_{1}(E-a^{-})\;+\;\frac{1}{2}\,\mathcal{N}_{1}(E-a^{+}).
\]
\end{thm}

To prove this result we need the following lemma:

\begin{lem}\label{lem:dos} For the integrated density of states
$\mathcal{N}_{1}$ of $H_{\text{\textup{per}}}+V_{\omega}$ we have
\[
\mathcal{N}_{1}(E)~=~\lim_{L\to\infty}\,\frac{1}{L}\;\mathbb{E}\left(N\left(\left(H_{\text{\textup{per}}}+V_{\omega}\right)_{M,L}^{X,Y}\right)\right)=~\lim_{L\to\infty}\,\frac{1}{L}\;\mathbb{E}\left(N\left(\left(H_{\text{\textup{per}}}+V_{\omega}\right)_{-L,-M}^{X,Y}\right)\right)
\]
for any fixed $M\in\mathbb{R}$ and any $X,Y\in\{D,N\}$.\end{lem}
\begin{proof} By the stationarity of the potential we have 
\begin{align*}
\mathbb{E}\left(N\left(\left(H_{\text{\textup{per}}}+V_{\omega}\right)_{M,L}^{X,Y}\right)\right)~=~\mathbb{E}\left(N\left(\left(H_{\text{\textup{per}}}+V_{\omega}\right)_{-(L-M)/2,\,(L-M)/2}^{X,Y}\right)\right)\,.
\end{align*}
Thus, the lemma follows from Theorem \ref{th:DOSexist}. \end{proof}

Now we prove Theorem \ref{th:DOScomp}.

\begin{proof} For $L>\left|M\right|$ we have 
\begin{multline}
\mathbb{E}\left(N\left(\left(H_{\omega}\right)_{-L,L}^{X,Y}\right)\right)\leqslant\mathbb{E}\left(N\left(\left(H_{\text{\textup{per}}}+U+V_{\omega}\right)_{-L,-M}^{X,N}\right)\right)+\\
+\mathbb{E}\left(N\left(\left(H_{\omega}\right)_{-M,M}^{N,N}\right)\right)+\mathbb{E}\left(N\left(\left(H_{\text{\textup{per}}}+U+V_{\omega}\right)_{M,L}^{N,Y}\right)\right).~~~~\label{eq:ineq}
\end{multline}
We take $M>0$ so large that $|U(x)-a^{-}|<\varepsilon/2$ for $x\leqslant-M$
and $|U(x)-a^{+}|<\varepsilon/2$ for $x\leqslant M$. Let us divide
inequality \eqref{eq:ineq} by $2L$. Then the middle term goes to
zero as $L\to\infty$. Moreover, in the limit the first term on the
right hand side can be bounded by $\frac{1}{2}\,\mathcal{N}_{1}(E-a^{-})\;+\;\varepsilon/2$.
Similarly, the third term can be bounded by $\frac{1}{2}\,\mathcal{N}_{1}(E-a^{+})\;+\;\varepsilon/2$.
Since $\varepsilon>0$ was arbitrary, we proved that 
\[
\limsup_{L\to\infty}\,\mathbb{E}\left(N\left(\left(H_{\omega}\right)_{-L,L}^{X,Y}\right)\right)~\leqslant~\frac{1}{2}\,\mathcal{N}_{1}(E-a^{-})\;+\;\frac{1}{2}\,\mathcal{N}_{1}(E-a^{+}).
\]

The opposite inequality is proved using the analogue of \eqref{eq:ineq}
for Dirichlet boundary conditions (instead of Neumann ones). \end{proof}

\end{document}